\documentclass[a4paper,11pt]{article}
\pdfoutput=1

\usepackage[utf8]{inputenc}
\usepackage[english]{babel}
\usepackage[T1]{fontenc}
\usepackage{lmodern}
\usepackage{microtype}
\usepackage[margin=1in]{geometry}
\usepackage{amsmath,amssymb,amsthm,mathtools}
\usepackage{enumitem}
\usepackage[numbers,sort&compress]{natbib}
\usepackage{xurl}
\usepackage[hidelinks]{hyperref}
\hypersetup{
  pdftitle={Complexity Amplification from Compression in Quantum Random Access Optimization},
  pdfauthor={Stuart Hadfield}
}

\newtheorem{theorem}{Theorem}[section]
\newtheorem{definition}[theorem]{Definition}
\newtheorem{corollary}[theorem]{Corollary}
\newtheorem{proposition}[theorem]{Proposition}
\newtheorem{lemma}[theorem]{Lemma}
\theoremstyle{remark}
\newtheorem{remark}[theorem]{Remark}

\newcommand{\QRAO}{\mathrm{QRAO}}
\newcommand{\QMC}{\mathrm{QMC}}
\newcommand{\OPT}{\mathrm{OPT}}
\newcommand{\NP}{\mathsf{NP}}
\newcommand{\StoqMA}{\mathsf{StoqMA}}
\newcommand{\BQP}{\mathsf{BQP}}
\newcommand{\QMA}{\mathsf{QMA}}

\usepackage{graphicx}
\usepackage{tikz}
\usetikzlibrary{arrows.meta,positioning,calc}
\title{Complexity Amplification from Compression in Quantum Random Access Optimization}
\author{Stuart Hadfield\\
\small USRA Research Institute for Advanced Computer Science,\\
\small Mountain View, California 94043, USA\\
\small \texttt{shadfield@usra.edu}\\
\small ORCID: \href{https://orcid.org/0000-0002-4607-3921}{0000-0002-4607-3921}}

\date{}

\begin{document}
\maketitle

\begin{abstract}
\begingroup

Compressed quantum encodings aim to overcome hardware limitations towards tackling challenging problems at scale, with many classical variables mapped onto noncommuting observables of fewer qubits. 
Despite their appeal, the tradeoffs and limitations of quantum compression remain poorly
understood. 
Classically, relaxations such as the semidefinite program formulation of MaxCut trade solution quality for computational efficiency. 
By contrast, quantum relaxations based on compression can \emph{amplify} the worst-case complexity of the problem being solved.
We study quantum random access optimization (QRAO), a
special case of the Pauli correlation encoding (PCE) framework that assigns up to three binary
variables to the Pauli $X$, $Y$, and $Z$ observables of each qubit, with the packing choices determining the compressed Hamiltonian to be optimized.
We identify explicit QRAO optimal
energy promise problems complete for $\NP$, $\StoqMA$, and $\QMA$, with
inverse-polynomial promise gaps for the latter two.
Existing results
establish hardness for abstract classes of Hamiltonians, but do not show these Hamiltonians
to be generated by valid QRAO packings. 
Our problem reductions do so while preserving inverse-polynomial promise gaps, without requiring gadgets or ancilla qubits. 
For any fixed packing, we show that weighted
MaxCut instances compress, up to a known shift and rescaling, to 
arbitrary nonnegative-weight pairwise Pauli couplings allowed by the packing.
A coupling using a single Pauli type is called aligned. 
For QRAO, using one aligned axis gives an $\NP$-complete energy problem. Under a
fixed two-to-one packing, positive single-variable terms along a second axis give a $\StoqMA$-complete transverse-field problem. Using two or three positive 
aligned Pauli axes gives $\QMA$-complete problems on unrestricted interaction graphs, while their bipartite restrictions lie in both $\BQP$ and $\StoqMA$.

We show that this computational hardness survives compilation and
is practically relevant. 
For any degree-ordered greedy QRAO compiler with fixed-order tie breaking, we construct dense, connected, regular, non-bipartite MaxCut families that assign exactly $d\in\{2,3\}$ variables to every qubit.
On these families MaxCut is $\NP$-complete while the QRAO
energy problem is $\QMA$-complete. Notably, this result applies directly to the current QRAO compiler implementation in Qiskit Optimization 0.7.0, confirming our hardness results are not artifacts of 
contrived packing rules.

We also consider optimization over restricted quantum state spaces.
On connected three-axis instances, optimization
over encoded classical assignments, product states, or separable states is
$\NP$-complete, whereas unrestricted optimization is $\QMA$-complete. Constant-factor approximation of the unrestricted three-axis relaxation is $\NP$-hard. 
Altogether our results identify worst-case complexity barriers arising from quantum compression, while making no broad claims about typical cases, or the performance and trainability of algorithm pipelines that use it.

\endgroup
\end{abstract}

\section{Introduction}

Quantum optimization algorithms offer a new computational paradigm for
tackling challenging combinatorial optimization problems~\cite{AbbasEtAl2024}.
Even before accounting for the additional qubit overhead required for fault
tolerance, the limited qubit counts of current devices pose a fundamental
obstacle to scaling such methods toward problem sizes beyond the reach of
classical algorithms. This challenge has motivated the development of compressed
encodings that map many classical variables to quantum observables on a smaller
quantum register.

Quantum random access optimization (QRAO)~\cite{Fuller2024QRAO} trades qubit
count for noncommutativity. In the standard $d$-to-$1$ QRAO encodings,
$d\in\{1,2,3\}$ classical binary variables are assigned to different Pauli
observables of each qubit. As these observables do not commute, a quadratic
problem objective that would normally produce a diagonal Ising Hamiltonian under the
usual single-variable per qubit mapping~\cite{Lucas2014Ising,Hadfield2021Representation}
instead becomes in general a noncommuting quantum Hamiltonian
on fewer qubits.
The
resulting problem relaxation is optimized over quantum states and its 
measurements are decoded into a classical solution. Classically, convex
relaxations such as semidefinite programs are generally chosen to yield
efficiently computable bounds and approximate solutions.  
QRAO instead maps a classical optimization problem to a quantum Hamiltonian optimization, but with no guarantee the resulting problem will be any easier. 
On explicit worst-case families, we show that QRAO compression produces a $\QMA$-complete optimal-energy problem from an $\NP$-complete
classical input family.

Previously, Teramoto et al.~\cite{Teramoto2023Extensions} identified as future work the complexity of finding a
relaxed state whose energy exceeds the classical optimum under QRAO, and noted that exact
extremal-state computation is $\QMA$-hard for general local
Hamiltonians~\cite{KempeKitaevRegev2006}. This general
observation does not classify QRAO complexity. Compressed Hamiltonians generated by valid inputs and
packings occupy a constrained positive-coupling image whose complexity can
depend on the active Pauli axes, interaction topology, and admitted quantum state
domain. We determine these boundaries and prove $\QMA$-completeness for explicit
reachable families, including fully occupied packings produced by the stable
greedy compiler used in the Qiskit Optimization 0.7.0 software package.

We call the assignment of variables to physical
qubits and Pauli axes a \emph{packing}. Its \emph{grouping} records which
variables share each qubit, while the QRAO \emph{compiler} is the classical procedure that
derives the grouping and Pauli slot (type) assignments from the classical problem input. In the compiler class
studied here, the classical support graph (variables as vertices and nonzero quadratic terms as edges) is colored greedily in nonincreasing degree order. 
A compiler is \emph{stable} if its greedy tie breaking rule preserves the fixed input order. 
Here we consider primarily applications to quadratic unconstrained binary optimization (QUBO)  problems, for which MaxCut is a prototypical $\NP$-hard problem. 
Each weighted problem clause is directly mapped,
up to a known shift and rescaling, to a corresponding Pauli coupling in the
compressed Hamiltonian.

QRAO is one of several qubit-efficient encoding strategies whose choices of
state space, loss function, and readout map create different resource
tradeoffs~\cite{TanEtAl2021,SawayaEtAl2023}. At the representation level, QRAO
is a special case of the more general Pauli-correlation encoding (PCE)
framework introduced by Sciorilli et al.~\cite{SciorilliEtAl2025}. PCE extracts
classical variable assignments as the signs of expectation values of Pauli
strings, with higher-locality strings offering stronger polynomial
compression. 
QRAO represents possible classical assignments by product states of single-qubit quantum random access code (QRAC) codewords, where each codeword encodes a definite assignment to the binary
variables packed into that qubit, 
and recovers classical solutions using QRAC-based rounding, such as random magic measurements or sign rounding of the assigned Pauli observables.
The input problem, together with the compression and compiler parameters, produces a unique quantum objective Hamiltonian, which is then optimized over all states or a specified class with a suitable approach such as parameterized quantum circuits or closely related variational approaches~\cite{AbbasEtAl2024,McCleanEtAl2016}. QRAO then recovers classical
solutions through repeated QRAC measurements and rounding, whereas general PCE methods
are typically paired with nonlinear correlator losses and sign decoding. Our
theorems classify the QRAO energy problem, not general PCE objectives.

The exactly characterizable QRAO setting nevertheless gives a concrete answer
to a broader question:
\begin{quote}
\emph{Can quantum compression result in a harder problem than the one we were given?}
\end{quote}

For QRAO we answer this question in the affirmative, up to the usual standard
beliefs in computational complexity theory, such as $\NP$ being unequal to $\QMA$,
its quantum witness counterpart. 
Our complexity classification is based on the
choice of compressed encoding alone and precedes the choices of ansatz,
training procedure, or decoder. 
We emphasize that our results primarily concern the global energy optimization problem over different classes of quantum states. We do not analyze or restrict to a particular quantum circuit ansatz or training method. 
Indeed, for a Pauli Hamiltonian with polynomially bounded coefficients in \(\ell_1\) norm, the energy of an efficiently prepared state can be estimated with high probability to inverse-polynomial additive accuracy using polynomially many Pauli observable measurements~\cite{McCleanEtAl2016}.

Considerations such as problem locality do not resolve the above question either.
The Local Hamiltonian quantum energy estimation decision problem is
$\QMA$-complete in general, while even when restricted to two-qubit interactions
it exhibits polynomial-time, $\NP$-, $\StoqMA$-, and $\QMA$-complete regimes
~\cite{KitaevShenVyalyi2002,KempeKitaevRegev2006,
BravyiBessenTerhal2006,CubittMontanaro2016,PiddockMontanaro2017,
MarwahaSud2026}. Hence existing results do not immediately imply hardness
classifications for QRAO. Indeed, a QRAO compiler is not able to produce
arbitrary interactions, with the set of possible compressed Hamiltonians
highly constrained by its rules for assigning variables to Pauli operators.
Our results therefore lie at the intersection of Hamiltonian complexity and
QRAO reachability.

We show families of connected, regular, exactly compressed positive-weight
MaxCut instances for which the classical threshold decision problem is
$\NP$-complete, but the corresponding QRAO energy promise problem
is $\QMA$-complete to inverse-polynomial precision. In between these extremes we
show a $\StoqMA$-complete variant that arises in compressing quadratic
QUBO problems, resulting when the encoding
choice leads to transverse-field Ising model Hamiltonians. We refer to this
change in completeness classification as \emph{complexity amplification from
quantum compression}.

Here and throughout, the QRAO energy problem is the standard
quantum Hamiltonian promise decision problem. For Hamiltonians on $n$ qubits, given thresholds $\alpha<\beta$ separated by at least an
inverse polynomial in $n$, the task is to distinguish whether the maximum eigenvalue 
$\lambda_{\max}(H_{\QRAO})\leq\alpha$ from the case 
$\lambda_{\max}(H_{\QRAO})\geq\beta$ under the promise that one holds. For convenience we primarily consider energy (objective) maximization instead of minimization throughout. Our
classical complexity results use the usual MaxCut threshold decision problem.
For complexity
statements, all input parameters are assumed to be rational, polynomially bounded in
magnitude, and specified using polynomially many bits in the problem size.

The hardness results used in our proofs classify abstract Hamiltonian
families. A QRAO Hamiltonian instead results as the image of a
classical objective function under a globally consistent
packing. One must therefore characterize the reachable set of Hamiltonians, map known
hard problem instances into this set while preserving their promise gaps, and
show that a QRAO compiler actually produces the hard Hamiltonians when
compressing. Proposition~\ref{prop:positive-cone}, Proposition~\ref{prop:exact-lift}, and
Theorem~\ref{thm:connected} below supply these three missing steps.

The central compiler statement is the following. Its two hardness claims
refer to one explicitly defined structural class, although different source
subfamilies may also witness them.

\begin{theorem}[Main result, informal]
\label{thm:informal-main}
For $d\in\{2,3\}$ and every fixed stable degree-ordered greedy
compiler, there is a class of connected, regular, non-bipartite positive-weight MaxCut
instances that is exactly $d$-to-$1$ compressed by the
compiler and has the
following properties.
\begin{enumerate}[label=(\roman*),itemsep=1pt]
\item The uncompressed classical MaxCut threshold decision problem is $\NP$-complete. \item The unrestricted QRAO energy promise problem with an
inverse-polynomial gap is $\QMA$-complete on the same input problem class.
\end{enumerate}
The construction works for every fixed tie-breaking order and requires no ancilla
qubits beyond the~$n$ qubits of the compressed Hamiltonian.
\end{theorem}

\paragraph{Main results.}

We classify QRAO along four dimensions, the first three of which any
complexity analysis of an analogous compressed encoding must also specify. These
are the set of reachable compressed Hamiltonians and their complexity, their
realization through an explicit problem and compiler, the
allowed compressed-state domain, and for QRAO the boundary between locally
alignable and nonalignable mixed Pauli axis interactions.
Figure~\ref{fig:quantum-roadmap} summarizes the representative energy problems and the state-domain classification.

\begin{figure}[tbp]
\centering
% Simplified from Figure 1 of the earlier Quantum manuscript.
% No arrows between interaction families: these are representative slices.
\resizebox{\linewidth}{!}{%
\begin{tikzpicture}[
  title/.style={anchor=west,font=\small\bfseries,text=black},
  card/.style={draw=black!45,rounded corners=2pt,align=center,
    font=\small,text=black,inner sep=5pt},
  interaction/.style={card,text width=4.38cm,minimum height=1.68cm},
  state/.style={card,font=\footnotesize,text width=3.04cm,minimum height=1.18cm},
  inclusion/.style={-{Stealth[length=1.7mm,width=1.2mm]},
    draw=black!65,line width=0.5pt}
]
\node[title] at (0,4.85) {(a) Representative energy problems under fixed packings};
\node[interaction,fill=gray!3] (one) at (2.39,3.62)
  {One aligned axis\\$d=1$ (no compression)\\[3pt]
   \textbf{\boldmath $\NP$-complete}\\[-1pt]
   {\footnotesize Theorem~\ref{thm:axis-classification}}};
\node[interaction,fill=gray!3] (field) at (7.68,3.62)
  {$ZZ$ couplings and $X$ fields\\Fixed $2$-to-$1$ packing\\[3pt]
   \textbf{\boldmath $\StoqMA$-complete}\\[-1pt]
   {\footnotesize Theorem~\ref{thm:stoqma-transverse}}};
\node[interaction,fill=gray!3] (axes) at (12.97,3.62)
  {Two or three aligned axes\\Fixed positive interaction\\[3pt]
   \textbf{\boldmath $\QMA$-complete}\\[-1pt]
   {\footnotesize Theorem~\ref{thm:axis-classification}}};

\node[title] at (0,2.13) {(b) State domains on one connected class after greedy compilation};
\node[anchor=west,font=\footnotesize,text=black] at (0,1.72)
  {Exactly $3$ variables per qubit; Theorems~\ref{thm:connected} and~\ref{thm:states}};
\node[state,fill=gray!3] (code) at (1.72,0.77)
  {QRAC codewords\\[3pt]\textbf{\boldmath $\NP$-complete}};
\node[state,fill=gray!3] (prod) at (5.69,0.77)
  {Pure product states\\[3pt]\textbf{\boldmath $\NP$-complete}};
\node[state,fill=gray!3] (sep) at (9.66,0.77)
  {Separable states\\[3pt]\textbf{\boldmath $\NP$-complete}};
\node[state,fill=gray!3] (all) at (13.63,0.77)
  {All quantum states\\[3pt]\textbf{\boldmath $\QMA$-complete}};
\draw[inclusion] (code.east) -- (prod.west);
\draw[inclusion] (prod.east) -- (sep.west);
\draw[inclusion] (sep.east) -- (all.west);
\end{tikzpicture}%
}
\caption{Selected QRAO complexity results.
(a) Three representative energy problems from fixed QRAO packings on
unrestricted interaction graphs.
(b) State-domain complexity on the connected, regular, non-bipartite
positive-weight MaxCut class $\mathcal C_3$ produced by stable greedy compilation (defined in Section~\ref{sec:compiler-realization}).
Arrows denote inclusion of state domains. QRAC codewords reproduce the
classical objective. Different source subfamilies are used to establish hardness for the different state domains.}
\label{fig:quantum-roadmap}
\end{figure}

\begin{enumerate}[leftmargin=*,itemsep=2pt]
\item \emph{Exact image and complexity spectrum.}

For a fixed (prescribed) packing, MaxCut with positive real weights realizes exactly the
cone of entrywise-nonnegative Pauli-coupling matrices
\begin{equation}
 -\frac d2\sum_{u<v}\sum_{P,Q}J_{uv}^{PQ}P_uQ_v,
 \qquad J_{uv}^{PQ}\geq0,
 \label{eq}
\end{equation}

where $P,Q\in\{X,Y,Z\}$ are Pauli matrices. We call such a two-qubit
interaction term aligned when $P=Q$, and a Hamiltonian aligned when every
two-qubit interaction term is aligned, although different terms may use
different Pauli axes. We consider the QRAO energy promise
problem defined below across the different cases. We show that for aligned single-axis compression
the problem is $\NP$-complete, a restricted two-to-one QRAO encoding with aligned
interactions along one axis and local fields along a second is
$\StoqMA$-complete, and for every fixed positive aligned interaction with two or three
Pauli axes the problem on unrestricted interaction graphs is $\QMA$-complete. The last family lies in both $\BQP$ and $\StoqMA$ when its interaction graph
is bipartite. The encoding therefore exposes $\NP$, $\StoqMA$, and $\QMA$
regimes without requiring an effective Hamiltonian approximation or gadgets. 

\item \emph{Compiler-level realization.}
For any fixed stable degree-ordered greedy compiler and fixed
tie-breaking order, we construct connected, regular, non-bipartite MaxCut instances with positive weights whose compiled packing is exactly full (i.e., every
qubit has $d$ variables mapped to it). For our hardness reductions we add small
weight edges to the classical support graph to force the compiler to produce the required
grouping of classical variables into qubits and their required assignments to
Pauli operators, while ensuring that the added interactions are sufficiently
small to preserve the original promise gap. Thus $\QMA$-hardness survives the
compiler stage rather than relying on a favorable packing chosen outside of it.
This compiler class includes the QRAO
encoder implemented in Qiskit Optimization 0.7.0 under the NetworkX 3.6.1
package's default degree-ordered greedy
coloring~\cite{QiskitQRAO2025,NetworkX2025}.

\item \emph{Quantum states, approximation, and compression factor.}
Within the same class of connected QRAO instances, 
we show that codeword,
product-state, and separable optimization are all $\NP$-complete, whereas
unrestricted optimization is $\QMA$-complete. Even for classically trivial
matching instances, entangled states increase the optimal QRAO relaxation value
over that of product states by exact factors $1$, $3/2$, and $2$ for one-,
two-, and three-axis encodings, respectively. Constant-factor approximation of
the three-axis quantum relaxation remains $\NP$-hard. 
Producing $n$ greedy color classes with $d$ variables each requires
$\Omega(dn^2)$ edges, and for fixed $d$ our construction uses within $O(n)$ edges of the exact lower bound
(Proposition~\ref{prop:density}).

\item \emph{A cycle boundary beyond aligned interactions.}

For isotropic mixed-axis interactions
(equal coefficients on the $XX$, $YY$, and $ZZ$ terms before local relabeling)
related by
cyclic Pauli permutations, we determine in linear time whether one can
independently relabel the Pauli axes on each qubit so that every interaction
becomes aligned. This is possible unless some graph cycle imposes contradictory
axis labels on the same qubit, and in that case such a cycle provides an
explicit certificate that simultaneous alignment is impossible.
Proposition~\ref{prop:twists} gives the formal criterion, its
linear-time algorithm, and the proof that a violating cycle rules out
alignment even by arbitrary local one-qubit unitaries.
\end{enumerate}

Graph topology supplies another boundary. A Hamiltonian is \emph{curable} if local single-qubit basis changes can make it stoquastic~\cite{MarvianLidarHen2019,KlassenTerhal2019}. 
On a bipartite interaction graph, a common local basis change makes the fixed aligned interaction stoquastic on every edge. 
For fixed positive aligned interactions, a further global permutation of the axes gives a Lee--Yang Hamiltonian. The recent gap result of Rayudu and Takahashi then places the bipartite energy problem in $\BQP$~\cite{RayuduTakahashi2026}.
The simplest genuinely compressed aligned encoding uses two noncommuting Pauli axes and
already enters into the $\QMA$-complete interaction regime. The $\StoqMA$ slice we identify shows
that compression into noncommuting observables does not lead only to the two extremes; adding
a sign-curable transverse field produces an intermediate quantum witness
class.

\paragraph{Relation to prior work.}
QRAO is a qubit-efficient encoding and rounding framework
for optimization introduced in Ref.~\cite{Fuller2024QRAO}. 
Subsequent work has studied approximation guarantees,
entanglement, noise, recursive methods, implementations, and hardware
experiments
~\cite{Teramoto2023Extensions,Teramoto2023Entanglement,Kondo2025Recursive,
Tamura2024Noise,Matsuyama2024Branch,He2025Nonvariational,
SharmaLau2026Hardware}. Here we discuss the complexity of the Hamiltonian
optimization problem created by compression, independently of the ansatz, training
method, or decoder; we provide some discussion of broader algorithmic and resource aspects in Section~\ref{sec:discussion}. 
A companion study asks whether an algorithm can efficiently prepare a state whose mean QRAO energy reaches the input graph's classical MaxCut optimum. For the standard $d=2,3$ encodings, a uniform polynomial-time procedure achieving this threshold on every input with inverse-polynomial preparation success would imply $\NP\subseteq\BQP$, even though a suitable QRAC product state exists for each input~\cite{HadfieldApproxQOpt2026}.
Our results leverage the fixed-interaction complexity classifications of Piddock and Montanaro
~\cite{PiddockMontanaro2017} as our primary hardness source. The important new analysis, techniques, and results are the determination of the exact
QRAO Hamiltonian image, the promise-preserving axis and transverse-field lifts, the
order-robust compiler construction, the density and cycle boundaries, and then
their combination into compression-induced $\NP$, $\StoqMA$, and $\QMA$
regimes. The complexity classification for optimization over product states follows from Ref.~\cite{KallaugherEtAl2025}. Our contribution is to realize this result within the same connected, compiler-produced class of QRAO instances used to classify optimization over QRAC codeword, separable, and unrestricted quantum states.

Closely related work considers Quantum MaxCut, the canonical $\QMA$-complete quantum analogue of
MaxCut~\cite{PiddockMontanaro2017,GharibianParekh2019}. Equation~\eqref{eq:qrao-qmc}
makes explicit that QRAO with isotropic three-axis compression can output instances of Quantum MaxCut. 
We do not claim this Hamiltonian identification alone as
new. Known complexity results do not by themselves show which Quantum MaxCut or more
general Hamiltonians are reachable from valid QRAO inputs and compiler
outputs. Our new results are the complete positive coupling image for
arbitrary fixed packings, the promise-preserving lifts, the compiler
realization, and the resulting complexity and state-domain classifications. 
Finally, Marwaha and Sud recently sharpened the classification of 2-local positive-weight Hamiltonians by identifying further $\StoqMA$-complete regions along with a new complexity phase boundary~\cite{MarwahaSud2026}. Rayudu and Takahashi subsequently proved a uniform-field spectral gap for Lee--Yang Hamiltonians and a polynomial-time adiabatic quantum algorithm for their ground-state energy estimation~\cite{RayuduTakahashi2026}. As shown in Proposition~\ref{prop:bipartite-stoqma}, this strengthens the containment of the bipartite positive aligned QRAO sector from $\StoqMA$ alone to both $\BQP$ and $\StoqMA$. Marwaha and Sud's classification concerns abstract
symmetric interactions, whereas our Hamiltonian image and compiler theorems instead identify the
reachable cone of QRAO Hamiltonians and show that hard subfamilies survive the compiler stage.

The paper is organized as follows. Section~\ref{sec:image} derives the compressed Hamiltonian image under QRAO. Section~\ref{sec:axis}
proves the $\NP$, $\StoqMA$, and $\QMA$ classifications. Section~\ref{sec:compiler}
gives the explicit compiler construction and density bound. Section~\ref{sec:states}
classifies complexity over restricted state domains, and Section~\ref{sec:approx} shows that approximation
hardness transfers through. We discuss the scope of our results in Section~\ref{sec:discussion}. 
\section{The Hamiltonian image of QRAO compression}
\label{sec:image}

\subsection{Positive MaxCut and QRAC codewords}

Consider the MaxCut problem. Let $G=(V,E,w)$ be a positive-weight graph. Using signs (bits) ~$x_i\in\{\pm1\}$, we write its objective as
\begin{equation}
 c_G(x)=\frac12\sum_{(i,j)\in E}w_{ij}(1-x_ix_j),
 \qquad w_{ij}>0.
 \label{eq:maxcut}
\end{equation}
The associated MaxCut threshold decision problem takes as input $G$ and a
rational threshold~$T$ and asks whether
$\max_{x\in\{\pm1\}^{|V|}}c_G(x)\geq T$.

For fixed $d\in\{1,2,3\}$, a standard $d$-to-$1$ quantum random access code assigns
each sign to a Pauli observable $P_i\in\{X,Y,Z\}$.  
For each bit string the corresponding QRAC codeword is the single-qubit quantum state whose
Bloch-vector components encode the $d$ assigned signs,
\[
\rho(x)=\frac{1}{2}\left(
I+\frac{1}{\sqrt d}\sum_{i=1}^{d}x_iP_i
\right),
\]
which satisfies
\begin{equation}
 \langle P_i\rangle=\operatorname{Tr}\!\left[\rho(x)P_i\right]
=\frac{x_i}{\sqrt d}.
\label{eq:qrac-codeword}
\end{equation}
One qubit can encode at most three signs with recovery probability greater than
one half~\cite{AmbainisEtAl2002,HayashiEtAl2006}. If the endpoints of a
classical edge occupy different physical qubits, the QRAO edge term is
~\cite{Fuller2024QRAO}
\begin{equation}
 h_{ij}^{(d)}=\frac{w_{ij}}2\bigl(I-dP_iP_j\bigr).
 \label{eq:qrao-edge}
\end{equation}
On a product of QRAC codewords this gives~\cite{Fuller2024QRAO,Teramoto2023Entanglement}
\begin{equation}
 \langle h_{ij}^{(d)}\rangle=\frac{w_{ij}}2(1-x_ix_j).
 \label{eq:qrao-codeword-edge}
\end{equation}
Summing over edges and comparing with
Eq.~\eqref{eq:maxcut} shows that QRAO codewords embed classical assignments
while preserving their objective values exactly.

We write $A_1=\{Z\}$, $A_2=\{X,Z\}$, and $A_3=\{X,Y,Z\}$ for the active Pauli axes
of the standard packings. A fixed packing is \emph{valid} when no
classical edge has both endpoints inside one packed block.

\begin{proposition}[Positive-coupling Hamiltonian image cone]
\label{prop:positive-cone}

Fix a full $d$-to-$1$ packing, where block $u$ consists of the
$d$ classical variables assigned to physical qubit $u$, with one variable
$x_{u,P}$ in each Pauli slot $P\in A_d$.
For $u<v$, let $J_{uv}^{PQ}\geq0$ be the total
weight of classical edges joining $x_{u,P}$ to $x_{v,Q}$. Then the compressed QRAO Hamiltonian is 
\begin{equation}
 \begin{aligned}
 H_{\QRAO}&=\frac{W_J}{2}I
 -\frac d2\sum_{u<v}\sum_{P,Q\in A_d}J_{uv}^{PQ}P_uQ_v,\\
 W_J&=\sum_{u<v,P,Q}J_{uv}^{PQ}.
 \end{aligned}
 \label{eq:positive-cone}
\end{equation}
Conversely, every polynomially specified collection
$\{J_{uv}\in\mathbb Q_{\geq0}^{d\times d}\}_{u<v}$ is realized exactly by a
positive-weight MaxCut instance with the same packing.
\end{proposition}

\begin{proof}
Summing Eq.~\eqref{eq:qrao-edge} over the classical edges gives
Eq.~\eqref{eq:positive-cone}. Conversely, for every nonzero entry
$J_{uv}^{PQ}$, add the classical edge $(x_{u,P},x_{v,Q})$ with that weight.
Every such edge joins different blocks, so the fixed packing is valid.
\end{proof}

The proposition identifies the complete positive pairwise (2-local) image of QRAO Hamiltonians. In
particular, QRAO does not produce only aligned interactions $P_uP_v$. A fixed
packing realizes every nonnegative matrix of Pauli pairings on each physical
edge.

\subsection{Aligned Pauli axis lifts}

The basic computational complexity amplification mechanism already appears in a simple aligned slice of
this Hamiltonian cone.

\begin{definition}[Aligned Pauli axis lift]
\label{def:aligned-lift}
Fix $d\in\{1,2,3\}$, a nonempty set $A\subseteq\{X,Y,Z\}$ with $|A|\leq d$,
and positive rational constants $c_P$ for $P\in A$. For every vertex $u$ of
$G$, make one classical variable $u_P$ for each $P\in A$. For every edge
$(u,v)\in E$ and $P\in A$, add the edge $(u_P,v_P)$ with weight
$c_Pw_{uv}$. Group the variables $\{u_P:P\in A\}$ on physical qubit $u$ and
assign $u_P$ to Pauli observable $P$.
\end{definition}

The classical graph is a disjoint union of $|A|$ weighted copies of $G$. The
copies become coupled only when corresponding variables are assigned to
noncommuting observables of the same qubit.

\begin{proposition}[Exact axis-lifting identity]
\label{prop:exact-lift}
Let $C_A=\sum_{P\in A}c_P$ and $W=\sum_{(u,v)\in E}w_{uv}$. The aligned lift
satisfies
\begin{equation}
 H_{\QRAO}=\frac12C_AWI-\frac d2\sum_{(u,v)\in E}w_{uv}
 \sum_{P\in A}c_PP_uP_v.
 \label{eq:exact-lift}
\end{equation}
Maximizing $H_{\QRAO}$ is therefore exactly equivalent, under a known affine
transformation, to minimizing
\begin{equation}
 K_A(G)=\sum_{(u,v)\in E}w_{uv}\sum_{P\in A}c_PP_uP_v.
 \label{eq:KA}
\end{equation}
\end{proposition}

\begin{proof}
Apply Eq.~\eqref{eq:qrao-edge} to every copied edge $(u_P,v_P)$ and sum over
$P\in A$ and $(u,v)\in E$.
\end{proof}

No approximation occurs in Proposition~\ref{prop:exact-lift}. Any gadgets or
interaction-strength range used by a source Local Hamiltonian reduction are
inherited, but the QRAO lift adds neither gadgets nor physical qubits. We use the aligned lift construction in the results below.

\section{NP, StoqMA, and QMA regimes of QRAO}
\label{sec:axis}
Here we use the standard Local Hamiltonian energy optimization promise (decision) problem~\cite{KitaevShenVyalyi2002,KempeKitaevRegev2006}, considered here in its equivalent 
maximum energy form. 
The Hamiltonian interaction
set and its coefficients are fixed constants. The input consists of a
positive rationally weighted graph $G$ and rational thresholds
$\alpha<\beta$ separated by an inverse polynomial in the encoded input length.
The task is to decide whether
$\lambda_{\max}(H_{\QRAO})\geq\beta$ or
$\lambda_{\max}(H_{\QRAO})\leq\alpha$, promised that one holds.

\subsection{NP to QMA computational complexity transition}

We now state our main theorem. 
\begin{theorem}[Pauli-axis classification]
\label{thm:axis-classification}
Let $r=|A|$ be the number of active axes in a fixed aligned lift with
$c_P>0$.
\begin{enumerate}[label=(\roman*),itemsep=1pt]
\item If $r=1$, the aligned QRAO energy promise problem is $\NP$-complete.
\item If $r\geq2$, the aligned QRAO energy promise problem is $\QMA$-complete.
\end{enumerate}
\end{theorem}

\begin{proof}

For $r=1$, the active-axis set is $A=\{P\}$, and
Eq.~\eqref{eq:exact-lift} becomes
\[
H_{\QRAO}
=\frac{cW}{2}I-\frac{dc}{2}
 \sum_{(u,v)\in E} w_{uv}P_uP_v
\]
for a fixed $c>0$. Choose a one-qubit unitary $U$ satisfying
$UPU^\dagger=Z$ and apply the same basis change to every qubit. The
conjugated Hamiltonian is diagonal 
\[
U^{\otimes |V|}H_{\QRAO}(U^\dagger)^{\otimes |V|}
=\frac{cW}{2}I-\frac{dc}{2}
 \sum_{(u,v)\in E} w_{uv}Z_uZ_v.
\]
Since $dc>0$, this is equivalent, up to
a known additive shift and positive rescaling, to the Hamiltonian $C_G$ for MaxCut~\cite{Hadfield2021Representation}, i.e. it satisfies $C_G|x\rangle =c_G(x)$ for every computational basis state $|x\rangle$, $x\in\{\pm1\}^{|V|}$. 
A maximizing state of $H_\QRAO$ can be chosen as a tensor product of
single-qubit~$P$ eigenstates. Measuring any maximizing state in this product
basis yields a
classical MaxCut witness for any aligned single axis~$P\in \{X,Y,Z\}$, and as discussed weighted MaxCut is
$\NP$-complete~\cite{GareyJohnsonStockmeyer1976}.

For $r\geq2$, absorb the common positive factor $d/2$ in
Eq.~\eqref{eq:exact-lift} into the edge weights and set $J_a=c_a$ on active
axes and $J_a=0$ otherwise. The resulting fixed Hamiltonian interaction is
\begin{equation}
 J_xX\!X+J_yY\!Y+J_zZ\!Z,
\end{equation}
where $J_a\geq0$ and at least two coefficients are positive. Piddock and
Montanaro prove that the positive weight Local Hamiltonian problem for this
interaction, with no one-local terms, is $\QMA$-complete precisely throughout
this parameter region~\cite[Theorem~2(i)]{PiddockMontanaro2017}. The exact affine
identity in Proposition~\ref{prop:exact-lift} transfers the required promise gap. The source
reductions use polynomially bounded weights and thresholds representable with polynomially many bits,
and the lifted QRAO instance preserves these properties. Explicitly,
if the source minimum-energy thresholds are $A_0<B_0$ and
$C=C_AW/2$, the QRAO maximum-energy thresholds are
$C-dA_0/2$ and $C-dB_0/2$, separated by
$d(B_0-A_0)/2$. Membership in $\QMA$ follows because
$-H_{\QRAO}$ is a polynomially specified two-local Hamiltonian.
\end{proof}

The $r=1$ case occupies one slot per qubit and so does not compress. Activating a
second Pauli axis both compresses two classical variables into one qubit and
can lead to noncommuting terms in the compressed Hamiltonian. The resulting energy problem can then require a quantum instead of classical witness, giving the simplest explicit example of our complexity amplification principle in the paper.
This is a worst-case classification of the optimal relaxation value, not a universal claim that
every multi-axis (compressed) instance is hard.

\subsection{Stoquastic boundaries}

A Hamiltonian is \emph{stoquastic} in a basis when all off-diagonal matrix
elements in that basis are real and nonpositive. The corresponding
stoquastic Local Hamiltonian promise problem is complete for $\StoqMA$
~\cite{BravyiBessenTerhal2006}. QRAO realizes this intermediate class in a
natural field-augmented slice.

\begin{theorem}[Transverse-field $\StoqMA$ completeness]
\label{thm:stoqma-transverse}
Fix a two-to-one packing with signs labeled $z_u,x_u$ each assigned to
$Z_u,X_u$ on physical qubit $u$. For positive rational $J_{uv}$ and $h_u$,
consider the objective function 
\begin{equation}
 c_{\rm TF}(z,x)=\frac12\sum_{(u,v)\in E}J_{uv}(1-z_uz_v)
 +\frac12\sum_u h_u(1-x_u).
 \label{eq:tf-qubo}
\end{equation}
At inverse-polynomial additive precision, the associated QRAO
energy promise problem is $\StoqMA$-complete. The classical
threshold decision problem remains $\NP$-complete.
\end{theorem}

\begin{proof}
The compressed QRAO Hamiltonian is
\begin{equation}
 H_{\rm TF}=\frac{W_J+W_h}{2}I
 -\sum_{(u,v)\in E}J_{uv}Z_uZ_v
 -\frac1{\sqrt2}\sum_u h_uX_u,
 \label{eq:tf-qrao}
\end{equation}
where $W_J=\sum_{uv}J_{uv}$ and $W_h=\sum_u h_u$. Maximizing
Eq.~\eqref{eq:tf-qrao} is equivalent up to a simple affine transformation to minimizing
\begin{equation}
 K_{\rm TF}=\sum_{(u,v)\in E}J_{uv}Z_uZ_v
 +\frac1{\sqrt2}\sum_u h_uX_u.
\end{equation}
Conjugating every qubit by $Z$ reverses the sign of the $X$ terms and leaves
the $ZZ$ terms fixed, so the resulting Hamiltonian is stoquastic and the problem lies in
$\StoqMA$. Piddock and Montanaro show the positive-coefficient $\{ZZ,X\}$ Hamiltonian problem is
$\StoqMA$-complete~\cite[Theorem~5]{PiddockMontanaro2017}. They also observe that every transverse-field coefficient may be taken
nonnegative because conjugating an individual qubit by $Z$ flips its $X$ term
without changing any $ZZ$ term. Thus our coefficients $J_{uv}>0$ and
$h_u/\sqrt2>0$ follow their convention.
Given a source field $g_u\geq0$ and promise gap $\Delta$, choose positive rational
$h_u$ with $|h_u/\sqrt2-g_u|\leq\Delta/(8|V|)$. The operator-norm error is at
most $\Delta/8$, so the promise remains inverse polynomial. 
Additional promise gap bookkeeping details are given in Appendix~\ref{app:promise}.
Hence, Eq.~\eqref{eq:tf-qrao} transfers the thresholds exactly up to this controlled
approximation. 

Classically, every $x_u$ optimizes independently to $-1$, while
the remaining objective is weighted MaxCut. Hence the classical threshold
problem is again $\NP$-complete.
\end{proof}

The transverse-field construction in Theorem~\ref{thm:stoqma-transverse}
requires the specified pairing of each $x_u$ with its corresponding $z_u$.
The variables encoded in the transverse field have no quadratic edges, so a graph-coloring compiler need not
place each $x_u$ beside its intended $z_u$ without additional control.
We do not establish $\StoqMA$-completeness under greedy compilation.

Graph structure provides a second stoquastic complexity boundary within the same 2-local Hamiltonian 
family.

\begin{proposition}[Bipartite aligned \texorpdfstring{$\BQP$ and $\StoqMA$}{BQP and StoqMA} containment]
\label{prop:bipartite-stoqma}
Fix a positive aligned interaction with at least two active Pauli axes. If the QRAO Hamiltonian interaction graph is bipartite, then its aligned QRAO energy promise problem
lies in $\BQP\cap\StoqMA$. Consequently,
$\QMA$-hardness of this restriction would imply $\StoqMA=\BQP=\QMA$.
\end{proposition}

\begin{proof}
We use the bipartite basis change observed by Piddock and Montanaro in their
lattice discussion~\cite[Sec.~1.1]{PiddockMontanaro2017}.
Conjugate every qubit on one side of the bipartition by the inactive Pauli in
the two-axis case, or by any Pauli in the three-axis case. The conjugation
flips two coefficients on every edge. Write the conjugating axis as $R$ and
the other two axes as $P,Q$. Up to a permutation of coordinates, the
coefficient triple becomes $(J_R,-J_P,-J_Q)$. In the two-axis case $R$ is
inactive and $J_R=0$. In the three-axis case $R$ may be any active axis. In
both cases the sum of the two flipped coefficients is
$-(J_P+J_Q)<0$. For an interaction
$J_xX\!X+J_yY\!Y+J_zZ\!Z$, Piddock and Montanaro's $\QMA$ parameter region requires all three
pairwise sums $J_x+J_y$, $J_x+J_z$, and $J_y+J_z$ to be positive
~\cite[Theorem~2]{PiddockMontanaro2017}.

Ref.~\cite{PiddockMontanaro2017} therefore places the
transformed interaction in $\StoqMA$, giving containment in this complementary
region but not completeness. (This is of course  consistent with the standard inclusions
$\NP\subseteq\mathsf{MA}\subseteq\StoqMA$~\cite{BravyiBessenTerhal2006}.)
Local basis changes preserve the spectrum, and
Proposition~\ref{prop:exact-lift} transfers the promise as required.

For $\BQP$ containment, relabel the axes so that
$J_x=\max\{J_x,J_y,J_z\}$. Conjugating one side of the bipartition by $Z$
and then globally interchanging the $X$ and $Z$ axes maps the fixed interaction
to
\begin{equation}
 -J_xZZ+J_zXX-J_yYY. \label{eq:lee-yang-bipartite}
\end{equation}
With positive edge weights, each edge has zero cross-couplings and satisfies
the Lee--Yang condition $J_x\geq\max\{|J_z|,|-J_y|\}$.
Rayudu and Takahashi prove that adding a uniform field of
strength $h>0$ along the dominant axis gives a nondegenerate ground state with
spectral gap at least $h/4$, and consequently a polynomial-time adiabatic
quantum algorithm for inverse-polynomial ground-state energy estimation
~\cite{RayuduTakahashi2026}.

Let $\Delta$ be the promise gap and $n$ the number of
physical qubits. Choosing $h=\Delta/(8n)$ changes the ground energy by at most
$nh=\Delta/8$, while the field-perturbed Hamiltonian has inverse-polynomial
gap. Estimating its ground energy to additive error at most $\Delta/8$ therefore
decides the original promise problem. The local basis
changes and the affine transformation in Proposition~\ref{prop:exact-lift}
are efficiently computable, proving containment in $\BQP$.
Since $\BQP,\StoqMA\subseteq\QMA$, $\QMA$-hardness of the restriction
would imply $\StoqMA=\BQP=\QMA$.

\end{proof}

The curing basis can also be exhibited directly. After a common permutation
of the axes, write the fixed interaction as
$J_xXX+J_yYY+J_zZZ$ with $J_x\geq J_y\geq0$. For two active axes,
take $J_z=0$. Conjugating the qubits in one part of the bipartition by $Z$
gives
\[
-J_xXX-J_yYY+J_zZZ.
\]
Its only potentially nonzero off-diagonal entries in the computational basis
are the nonpositive values $-(J_x-J_y)$ between $|00\rangle$ and $|11\rangle$,
and $-(J_x+J_y)$ between $|01\rangle$ and $|10\rangle$.
As every edge carries a positive multiple of the same fixed interaction,
this basis cures all edges at once. The argument applies to the minimum energy
Hamiltonian $K_A(G)$ of Proposition~\ref{prop:exact-lift}. The negative sign
in Eq.~\eqref{eq:exact-lift}
then converts its energy promise to the QRAO maximization convention.

Thus, for the QRAO Hamiltonian interaction graph an odd cycle is necessary for $\QMA$-hardness within the positive aligned
sector, unless $\StoqMA=\BQP=\QMA$. This does not weaken
Theorem~\ref{thm:axis-classification}, whose unrestricted $\QMA$-hardness is
witnessed by non-bipartite source families.

\begin{remark} \label{rem:stoqma-signed}
One-local terms are unnecessary if we extend the input problem class to QUBO problems, which includes arbitrary signs of edge terms, and so contains the positive weight MaxCut class used elsewhere
in the paper as a proper subset.
Under a fixed full three-to-one packing, assign coefficients
$(+w_{uv},-w_{uv},+\gamma w_{uv})$ to the $(X,Y,Z)$ copies of each positive
edge, with fixed rational $\gamma>1$. The coefficient map in Eq.~\eqref{eq:positive-cone}, extended to signed weights, gives
the QRAO image $-\tfrac32K_\gamma$ for
$K_\gamma=\sum_{uv}w_{uv}(X_uX_v-Y_uY_v+\gamma Z_uZ_v)$, up to an
irrelevant identity shift. Maximizing the QRAO image is therefore equivalent
to minimizing $K_\gamma$, whose ground-state energy problem is
$\StoqMA$-complete. Indeed, the ``Furthermore'' clause of
Piddock and Montanaro's Theorem~2 states $\StoqMA$-completeness when
$\alpha=-\beta\ne0$, $\alpha+\gamma>0$, and $\beta+\gamma>0$; taking
$\alpha=1$, $\beta=-1$, and $\gamma>1$ gives exactly $K_\gamma$
~\cite[Theorem~2]{PiddockMontanaro2017}.
\end{remark}

Theorems~\ref{thm:axis-classification} and
~\ref{thm:stoqma-transverse}, together with Proposition~\ref{prop:bipartite-stoqma}, expose
$\NP$, $\StoqMA$, and $\QMA$ regimes within QRAO, while identifying
which conclusions depend on fields, signs, or graph topology.

\subsection{Canonical examples with equal axis coefficients}
Here the coefficient is the same on every active Pauli axis, while the edge weights $w_{uv}$ may vary. We give two such QRAO Hamiltonian families.

First, for $d=2$ consider the aligned Hamiltonian 
\[
H_{\QRAO}
=
WI-\sum_{(u,v)\in E}w_{uv}(X_uX_v+Z_uZ_v).
\]
A common single-qubit basis change applied across all qubits maps it to or from the antiferromagnetic
XY model, up to the identity shift and the sign arising from our
maximum-energy convention. Consequently, the positive two-axis
QRAO energy promise problem with equal axis coefficients is $\QMA$-complete~\cite{PiddockMontanaro2017}, 
confirming a special case of Theorem~\ref{thm:axis-classification} above.

Second, for $d=3$ consider the Quantum MaxCut  Hamiltonian~\cite{GharibianParekh2019}
\begin{equation}
 H_{\QMC}(G)=\frac12\sum_{(u,v)\in E}w_{uv}
 \left(I-X_uX_v-Y_uY_v-Z_uZ_v\right).
 \label{eq:qmc}
\end{equation}
Its maximum energy promise problem is $\QMA$-complete by the
positive-weight antiferromagnetic Heisenberg model result of
Piddock and Montanaro~\cite[Theorem~2]{PiddockMontanaro2017}.
Note that some Quantum MaxCut conventions place $1/4$ rather than $1/2$ in front of
each edge term.
Our normalization matches the MaxCut and QRAO codeword
conventions above, and differs only by a possible overall factor of two.
For QRAO we have the exact identity
\begin{equation}
 H_{\QRAO}(G^{\uparrow3})=3H_{\QMC}(G),
 \label{eq:qrao-qmc}
\end{equation}
where $G^{\uparrow 3}$ denotes the fixed three-axis lift of the classical input graph $G$,\footnote{$G^{\uparrow 3}$ has variables $u_X,u_Y,u_Z$, packed onto qubit $u$, and replaces each edge $(u,v)$ by $(u_X,v_X)$, $(u_Y,v_Y)$, and $(u_Z,v_Z)$ of the same weight. Thus $G^{\uparrow 3}$ is the classical support graph, whereas the compressed Hamiltonian interaction graph is $G$.} and so the QRAO optimization problem is once again seen to be
$\QMA$-complete. 
In contrast, as explained above, restricting the same Hamiltonian to QRAC
codewords recovers classical MaxCut. As $G^{\uparrow3}$ consists of three disjoint copies of $G$, its
optimal codeword value is $3\cdot \operatorname{MaxCut}(G)$.

Hence QRAO contains familiar $\QMA$-complete spin models as exact compressed images,
not merely as low-energy effective theories, and so cannot avoid their computational complexity in the general worst-case. We leverage the Hamiltonians above in some of the results below.

\section{Compiler-layer hardness and existing software realization}
\label{sec:compiler}

The aligned Pauli axis lift fixes which variables share a physical qubit. An arbitrary QRAO
implementation instead may derive the packing from the classical support graph. We reserve \emph{Hamiltonian interaction graph} for the graph on physical qubits joined by nonzero two-qubit couplings.
We now show that a natural class of greedy QRAO compilers can be forced to output computationally hard compressed Hamiltonians. 

The edges needed to encode the source Hamiltonian also affect the compiler's
coloring. We include them in a regular support graph whose greedy coloring
yields the desired color classes, then assign the remaining edges positive
weights small enough to preserve the promise gap.

\begin{definition}[Stable degree-ordered greedy compiler]
\label{def:grouping-rule}
Fix a total order $\prec$ on the variables and $d\in\{2,3\}$.
We call a compiler stable degree-ordered greedy if it performs
the following steps.
\begin{enumerate}[label=(G\arabic*),itemsep=1pt]
\item Greedily color the vertices in nonincreasing degree order, preserving
$\prec$ inside each degree tie, and assign each vertex the smallest available
nonnegative color.
\item List every color class using the same order $\prec$.
\item Partition each list into consecutive blocks of at most $d$ variables and
assign the positions in a block to fixed distinct Pauli axes.
\end{enumerate}
\end{definition}

The definition separates the compiler into coloring, grouping,
and Pauli slot assignment. Theorem~\ref{thm:connected} below applies to any fixed implementation
of these three steps, with any fixed tie-breaking order. Concretely, we report that the
\texttt{QuantumRandomAccess\allowbreak Encoding} function in the publicly available Qiskit
Optimization 0.7.0 software package works by constructing this graph (called the quadratic interaction graph in the implementation), invoking 
NetworkX 3.6.1's default \texttt{largest\_first} greedy coloring, listing each
color class in increasing variable index order, and then packing consecutive blocks
into fixed Pauli slots~\cite{QiskitQRAO2025,NetworkX2025}. With increasing
variable index as the total order~$\prec$, this procedure is an exact instantiation of Definition~\ref{def:grouping-rule}. The correspondence follows at the source code level because Qiskit inserts all variable nodes in increasing
index order and sorts each color class before packing, while NetworkX's stable
degree sort preserves node-insertion order for degree ties. These statements concern the specified software implementations rather than a general API guarantee. We make no claims that other releases, preprocessing pipelines, or compiler heuristics
will follow the same rules or have the same consequences. 
The code and data archive~\cite{Hadfield2026QRAOCode} includes a version-checking regression script. It constructs the
classical support graph used in the proof of Theorem~\ref{thm:connected} and verifies that the compiler
produces the predicted full packing for $d=2,3$ using these exact
releases.

\subsection{Connected compiler realization with local-axis control} \label{sec:compiler-realization}
For $d=2$, $S_2$ denotes the two-element permutation group acting on
the active Pauli axes. 
For $d=3$, Theorem~\ref{thm:connected} below allows an independent cyclic
relabeling of the $X$, $Y$, and $Z$ slots at each qubit. Such
relabelings are implemented by sign-free single qubit Clifford unitaries, so they do not reverse the sign of any interaction coefficient, and positive QUBO weights remain positive under them. 
General Pauli-axis permutations are excluded because an odd
permutation necessarily introduces a sign flip on at least one active axis.
The allowed cyclic Pauli permutations are the identity
and the two nontrivial common shifts of $X,Y,Z$
that give the cyclic group of
order three $C_3 \subset S_3$, 
and are implemented by
powers of a sign-free single qubit Clifford (i.e., one that permutes the Pauli axes
without introducing minus signs). The relative cyclic offset between the Pauli
labels at the endpoints of an oriented edge is called its \emph{twist}.
The sign-free Clifford $C$ defined in Eq.~\eqref{eq:cyclic-clifford} of Appendix~\ref{app:twists} implements $X\mapsto Y\mapsto Z\mapsto X$.
Writing $k_u\in\mathbb Z_3$ for the local frame at qubit $u$ gives
$t_{uv}=k_v-k_u$ on an oriented edge. Hence a twist assignment comes from
consistent local frames exactly when its oriented sum vanishes modulo three
around every cycle. A spanning-tree traversal either recovers the frames or
returns a violating cycle in $O(|V|+|E|)$ time.

For each $d\in\{2,3\}$ and fixed stable degree-ordered greedy
compiler, let
$\mathcal C_d$ denote the structural class of positive-weight MaxCut
instances, together with their compiled (compressed) Hamiltonians obtained by the
support-and-packing construction below, without restricting
the source problem thresholds.

\begin{theorem}[Compiler realization theorem]
\label{thm:connected}
Fix $d\in\{2,3\}$ and a stable degree-ordered greedy compiler.
Let $H_{\mathrm{src}}$ be any positive aligned $d$-axis QRAO Hamiltonian on
$n\geq3$ qubits, with an inverse-polynomial promise gap $\Delta$. 
For every fixed variable order and every
choice of permutation 
$\gamma_r\in S_2$ for $d=2$ or cyclic
permutation $\gamma_r\in C_3$ for $d=3$, one can construct in polynomial time
a positive-weight MaxCut instance whose compiled Hamiltonian is
\begin{equation}
 H_{\mathrm{out}}=U_\gamma H_{\mathrm{src}}U_\gamma^\dagger
 +\alpha I+V_0,
 \qquad \lVert V_0\rVert\leq\Delta/8,
 \label{eq:controlled-completion}
\end{equation}
with $\lVert \cdot \rVert$ the spectral norm. 
Here $U_\gamma$ is the tensor product of local single-qubit
unitaries implementing the permutations $\gamma_r$, $\alpha$ is a known
constant, and $V_0$ is a nonidentity perturbation contributed by
filler edges.
The constructed classical support graph is connected, $(n+d-2)$-regular, and non-bipartite with 
\begin{equation}
 d\binom n2+\binom d2n
 \label{eq:edge-count}
\end{equation}
edges.
Every greedy color class contains exactly $d$ variables, so the compiler uses
exactly $n$ physical qubits for $dn$ classical variables. 
  \end{theorem}

\begin{proof}
Write the fixed order as
$x_0\prec x_1\prec\cdots\prec x_{dn-1}$. For each axis label
$a\in\{0,\ldots,d-1\}$, define
\begin{equation}
 V_a=\{x_{dr+\gamma_r(a)}:0\leq r<n\}
\end{equation}
and write $v_{a,r}=x_{dr+\gamma_r(a)}$. Make every $V_a$ a clique. For each
pair $a<b$, add the shifted perfect matching
\begin{equation}
 \{(v_{a,r},v_{b,r+1\bmod n}):0\leq r<n\}.
 \label{eq:shifted-matching}
\end{equation}
Each vertex has $n-1$ neighbors in its clique and one neighbor in every other
part. The support is therefore $(n+d-2)$-regular and has the edge count in
Eq.~\eqref{eq:edge-count}. The shifted matchings connect the parts, and every
part contains a triangle because $n\geq3$.

Let $R_r=\{v_{a,r}:0\leq a<d\}$ denote row $r$. Clique edges remain
inside one part, while every cross-part edge changes the row index by one
modulo $n$. Hence each $R_r$ is independent; the assumption $n\geq3$ ensures
that even the wraparound edges between rows $0$ and $n-1$ never stay within a
row.
Equal degree makes (G1) process vertices in the fixed order. We claim
that greedy coloring assigns color $r$ to every $v_{a,r}$. Inductively, the
earlier vertices in the clique $V_a$ already carry colors
$0,\ldots,r-1$, so all smaller colors are forbidden. 
Every cross-part neighbor has row index $r+1$ or $r-1$ modulo $n$, never $r$.
For $r=0$, the $r-1$ neighbor wraps to row $n-1$ and is still
unprocessed. For $r=n-1$, the $r+1$ neighbor wraps to row $0$ and already has
color $0$. In every other case, an already processed cross-part neighbor has
its different row color by induction.
If
not, it cannot affect the current choice. Thus color~$r$ is available and
is the greedy choice. Color class $r$ is consequently the ordered row
$(x_{dr},\ldots,x_{dr+d-1})$. It is independent, fills one block, and assigns
part $a$ to Pauli position $\gamma_r(a)$.

Now identify row $r$ with source qubit $r$. For each source edge $(u,v)$ and each
axis part $a$, put its positive source weight on the same-part edge
$(v_{a,u},v_{a,v})$. These edges realize
$U_\gamma H_{\mathrm{src}}U_\gamma^\dagger$.
Give every remaining edge a sufficiently small common
positive weight. As detailed in Appendix~\ref{app:compiler-completion}, its
total contribution has the form $\alpha I+V_0$, with
$\lVert V_0\rVert\leq\Delta/8$, while all weights retain polynomial bit
complexity. This completes the proof.
\end{proof}

The construction of the theorem is adapted to the fixed order of the compiler. It does not rely on
finding a favorable relabeling of a given instance. It is also robust to small
independent changes of positive edge weights. 
Indeed, 
each edge contributes
$w_e(I-dP_eQ_e)/2$ to the QRAO Hamiltonian, and so changing its weight by
$\xi_e$ contributes
$\xi_e(I-dP_eQ_e)/2$ to $\Delta H$. Since $P_eQ_e$ has eigenvalues
$\pm1$, we have $\lVert I-dP_eQ_e\rVert=d+1$, and the triangle inequality
then gives
\[
\lVert\Delta H\rVert
\leq
\frac{d+1}{2}\sum_e|\xi_e|.
\]
Reserving an additional $\Delta/8$ of slack at each output threshold
leaves a promise gap of at least $\Delta/2$ and preserves the answer whenever
$\sum_e|\xi_e|\leq\Delta/(4(d+1))$, provided all edge weights remain positive.
Within the space of positive edge weights on this fixed support, every constructed promise Hamiltonian instance
has an inverse-polynomial neighborhood in $\ell_1$
distance with the same answer and packing.

We now show that the edge density of the classical support graph is unavoidable under the exact color-class requirement.

\begin{proposition}[Density under exact color class sizes]
\label{prop:density}
Suppose (G1) is applied to $dn$ variables and every resulting color class has
size exactly $d$. Then the classical support graph has
\begin{equation}
   |E|\geq dn(n-1)/2=\Omega(dn^2).
\end{equation}
For $n\geq3$, the ratio between the edge count in Theorem~\ref{thm:connected} and this lower
bound is $(n+d-2)/(n-1)$, which tends to $1$ as $n$ grows with fixed $d$.
\end{proposition}
\begin{proof}
Orient each edge from its earlier to its later endpoint in the greedy order.
A vertex assigned color $k$ has earlier neighbors of every color
$0,\ldots,k-1$, and so has in-degree at least~$k$. 
Since there are $d$
vertices in each color class 
and the greedy coloring uses consecutive colors $0,\dots,n-1$, counting every edge once gives 
\[
 |E|=\sum_v\deg^-(v)
 \geq d\sum_{k=0}^{n-1}k
 =\frac{dn(n-1)}2.
\]
The bound is attained by $d$ disjoint copies of $K_n$, ordered in $n$ rows
with one vertex from each copy in each row.
\end{proof}

Proposition~\ref{prop:density} assumes that each color class has exactly
$d$ vertices, as in Theorem~\ref{thm:connected}. This is stronger than full
occupancy, since (G3) splits a class of size $kd$ into $k$ full blocks.
An edgeless graph on $dn$ vertices has just one color class but still fills
all $n$ qubits. Thus full occupancy alone does not imply quadratic
support-graph density.

\begin{corollary}[Compiler-level complexity amplification]
\label{cor:main}
For every $d\in\{2,3\}$ and every fixed stable degree-ordered
greedy compiler,
the classical positive-weight MaxCut threshold decision
problem is $\NP$-complete on the input instances in $\mathcal C_d$,
while the unrestricted compiled QRAO energy promise problem is
$\QMA$-complete on $\mathcal C_d$ with an inverse-polynomial gap.
\end{corollary}

\begin{proof}
Membership in $\NP$ for the classical problem and $\QMA$ for the quantum problem
is immediate. For classical hardness, start from unweighted MaxCut and place
$d$ identical copies of every source edge on the $d$ clique parts. Choose the
total filler weight below one quarter of the unit classical promise gap. The
heavy-edge optimum is $d\operatorname{MaxCut}(G)$, and the filler cannot close
the gap. Theorem~\ref{thm:connected} supplies the required classical support graph and packing. 
For quantum hardness, start from either two-axis or three-axis instances with equal axis coefficients
in the $\QMA$-complete region of Theorem~\ref{thm:axis-classification}.
Equation~\eqref{eq:controlled-completion} preserves the inverse-polynomial
promise gap. The two reductions use the same class
$\mathcal C_d$ though not necessarily the same individual instances.
\end{proof}

Appendix~\ref{app:twists} constructs a $\QMA$-hard three-axis family
whose non-filler source terms become mixed-axis under balanced local Pauli
frame choices (Corollary~\ref{cor:mixed-axis-source}).

\section{Relaxation optima and complexity across state domains}

\label{sec:states}

QRAO optimization can also be applied and studied over restricted classes of quantum states rather than the full compressed Hilbert space. 
Here we consider the nested state domains 
\begin{equation}
 \text{QRAC codewords}\subset\text{pure product states}
 \subset\text{separable states}\subset\text{all quantum states}.
 \label{eq:domains}
\end{equation}
Separable states are convex mixtures of product states.
For a nonempty compact family $\mathcal S$ of density operators, write
\begin{equation}
 \OPT_{\mathcal S}(H)=\max_{\rho\in\mathcal S}\operatorname{Tr}(H\rho).
\end{equation}
QRAC codewords reproduce the original MaxCut objective by
Eq.~\eqref{eq:qrao-codeword-edge}. In the isotropic three-axis case,
Eqs.~\eqref{eq:qrao-codeword-edge}  and~\eqref{eq:qrao-qmc} together give
\begin{equation}
 \begin{aligned}
 \OPT_{\mathrm{all}}(H_{\QRAO})&=3\cdot\OPT_{\QMC}(G),\\
 \OPT_{\mathrm{code}}(H_{\QRAO})&=3\cdot \operatorname{MaxCut}(G).
 \end{aligned}
 \label{eq:state-values}
\end{equation}
Consequently,
\begin{equation}
 \frac{\OPT_{\mathrm{all}}(H_{\QRAO})}
 {\OPT_{\mathrm{code}}(H_{\QRAO})}
 =\frac{\OPT_{\QMC}(G)}{\operatorname{MaxCut}(G)}.
 \label{eq:relaxation-gap}
\end{equation}
This ratio is the QRAO relaxation gap for the given instance. It plays the role
of an integrality gap, with QRAC codewords representing the discrete
feasible solutions and all quantum states forming the relaxed domain.
Equation~\eqref{eq:relaxation-gap} shows that it is equal to the
Quantum MaxCut to Classical MaxCut relaxation gap.

We now show this gap can be strict even when $G$ is a simple matching (for which every classical
edge can be cut).

\begin{proposition}[Exact matching gap]
\label{prop:matching}
Let $G$ be a matching of total edge weight $W$, and use the full
$d$-axis lift with equal axis coefficients for $d\in\{1,2,3\}$. Then
\begin{equation}
 \OPT_{\mathrm{code}}=\OPT_{\mathrm{prod}}=\OPT_{\mathrm{sep}}=dW,
 \qquad
 \OPT_{\mathrm{all}}=\frac{d(d+1)}2W.
 \label{eq:matching-gap}
\end{equation}
Thus unrestricted states improve on the optimal value by factors $1$, $3/2$, and $2$
for one, two, and three axes, respectively.
\end{proposition}

\begin{proof}
Every matching edge in every classical copy can be cut, so codewords attain
$dW$. For a product state, let $r_u^{(A_d)}$ and $r_v^{(A_d)}$ be the endpoint
Bloch vectors projected onto the active axes. Their norms are at most one, so
$r_u^{(A_d)}\mathbin{\cdot}r_v^{(A_d)}\geq-1$, with equality for antipodal
unit vectors supported on those axes. Product states therefore attain no more
than $dW$, and convex mixtures cannot improve this value. For unrestricted
states, distinct matching edges act on disjoint qubit pairs. A Bell state on
each pair has eigenvalue $-1$ for every active same-axis Pauli product, giving
$d(d+1)w/2$ per edge.
\end{proof}

The matching result separates relaxation strength from source hardness. It 
does not imply a better decoded cut. Indeed, every matching edge was already cut by a
codeword.
For $d=3$, the same example limits the worst-case product-state approximation ratio for Quantum MaxCut to $1/2$ of the unrestricted optimum. Parekh and Thompson achieve this guarantee on arbitrary positive-weight graphs~\cite{ParekhThompson2022Optimal}.

For a pure product state, let $r_u\in\mathbb R^3$ be the unit Bloch vector of
qubit $u$ (i.e., the vector of its Pauli expectation values for $X$, $Y$, and $Z$). On the isotropic three-axis slice, the Hamiltonian $K_A(G)$ defined
in Eq.~\eqref{eq:KA} has expectation
\begin{equation}
 \langle K_A(G)\rangle=
 \sum_{(u,v)\in E}w_{uv}\,r_u\mathbin{\cdot}r_v.
 \label{eq:product-energy}
\end{equation}
Kallaugher et al. prove that the corresponding product-state Quantum MaxCut
problem is $\NP$-complete even for positive unit-weight graphs
~\cite[Corollary~6]{KallaugherEtAl2025}. A polynomial-precision list of Bloch
vectors is an $\NP$ witness. Choosing the coordinate tolerance below a fixed
fraction of the promise gap divided by the coefficient $\ell_1$ norm controls
the energy error. A linear objective has the same maximum over convex mixtures
of product states as over pure product states, so
$\OPT_{\mathrm{sep}}=\OPT_{\mathrm{prod}}$ and the separable problem is also
$\NP$-complete.

\subsection{Stability and restricted state space complexity} 
\label{sec:stability}
The following stability bound shows how to control small perturbations, such as those used to enforce connected compiler output instances in Section~\ref{sec:compiler}. 
\begin{lemma}[Stability under perturbations] \label{lem:stability}
For every nonempty compact family $\mathcal S$ of density operators, and
Hermitian operators $H$ and $V$,
\begin{equation}
  \left|
\OPT_{\mathcal S}(H+V)-\OPT_{\mathcal S}(H)
\right|
\;\leq\; \lVert V\rVert .  
\end{equation}
\end{lemma}

\begin{proof}
For every (normalized) $\rho\in\mathcal S$ we have $|\operatorname{Tr}(V\rho)|
\leq\lVert V\rVert$, which for $H+V$ implies 
$\operatorname{Tr}((H+V)\rho)
\leq
\operatorname{Tr}(H\rho)+\lVert V\rVert$.
Maximizing over $\rho\in\mathcal S$ then gives
\[
\OPT_{\mathcal S}(H+V)
\leq
\OPT_{\mathcal S}(H)+\lVert V\rVert.
\]
Applying the same argument to
$H=(H+V)-V$ gives the reverse inequality.
\end{proof}
We apply the lemma to show QRAO optimization remains hard on restricted state spaces.

\begin{theorem}[State domain complexity classification]
\label{thm:states}
On the class $\mathcal C_3$ of connected, regular, non-bipartite, exactly
three-to-one compressed positive-weight MaxCut instances, the
promise problems for optimization over QRAC codewords, pure
product states, and separable states are $\NP$-complete with inverse-polynomial
gap, becoming $\QMA$-complete over the full Hilbert space.
\end{theorem}

The theorem does not say that restricting the quantum ansatz makes optimization easy.
Rather, it shows that admitting
unrestricted states can change the relevant completeness classification from
$\NP$ to $\QMA$.
The $\NP$-complete results are derived by considering different Hamiltonian subfamilies inside $\mathcal C_3$ for each restricted state domain.

\begin{proof}[Proof overview]
For QRAC codewords, the claim follows from
Eq.~\eqref{eq:qrao-codeword-edge} and the classical part of
Corollary~\ref{cor:main}. For product states, combine the known three-axis
$\NP$-hard family with Theorem~\ref{thm:connected} and
Lemma~\ref{lem:stability}; linearity gives the same optimum over separable
states. The unrestricted claim is the $d=3$ case of
Corollary~\ref{cor:main}. Full details are given in
Appendix~\ref{app:state-domain-proof}.
\end{proof}

\section{Hardness of Approximation}
\label{sec:approx}

The exact scaling in Eq.~\eqref{eq:qrao-qmc} also transfers multiplicative
value hardness. Piddock proves that there is a constant $\eta<1$ for which an
$\eta$-multiplicative approximation to Quantum MaxCut is $\NP$-hard even on
unweighted graphs of constant bounded degree
~\cite[Theorem~1]{Piddock2025QMC}. 
The result guarantees a constant $\eta<1$, independent of the instance
size, but does not provide an explicit approximation ratio suitable for
quantitative comparison.

For comparison, the results of Hwang et al.~\cite{HwangEtAl2023} assume both the Unique Games Conjecture and a vector-valued Borell conjecture. Heilman~\cite[Theorem~1.7 and Corollary~1.9]{Heilman2026Sharp} proves the needed inequality, giving sharp product-state hardness and an unrestricted Quantum MaxCut bound under Unique Games alone. The theorem below uses Piddock's unconditional result.

\begin{theorem}[Constant-factor approximation hardness]
\label{thm:approx}
There exists a constant $\eta<1$ for which approximating the isotropic three-axis
QRAO optimal energy within factor $\eta$ is $\NP$-hard on the union of three identical, disjoint,
unweighted bounded-degree MaxCut graphs. Every stable
degree-ordered greedy compiler realizes the required packing under the
vertex-major order. There is also
a constant $\eta'<1$ giving $\NP$-hardness on the connected, regular,
positive-weight, fully compressed class $\mathcal C_3$.
\end{theorem}

\begin{proof}[Proof overview]
Equation~\eqref{eq:qrao-qmc} preserves the source
multiplicative ratio on the three-copy lift. Applying
Theorem~\ref{thm:connected} with filler contribution sufficiently small
relative to the positive source threshold gives a fixed ratio
$\eta'<1$ on $\mathcal C_3$. Full details are given in
Appendix~\ref{app:approx-proof}.
\end{proof}

The connected family is weighted and dense, whereas the first family is
unweighted and bounded degree but disconnected before compression. The theorem
concerns approximation of the best relaxation value, not the quality of a
classical string produced by a particular decoder.
The first family also relies on the vertex-major input order to
produce its packing. In contrast, the connected family inherits the
any-fixed-order robustness of Theorem~\ref{thm:connected}, at the cost of
weights and a dense classical support graph.

For the uncompleted isotropic lift, Eq.~\eqref{eq:qrao-qmc} scales every
state's energy, and hence the optimum, by three. A multiplicative guarantee
for a Quantum MaxCut state therefore carries over to its QRAO relaxation value.
In the connected construction, the filler includes an identity shift that can
change these ratios. Appendix~\ref{app:approx-proof} bounds the full filler
contribution separately.

\section{Discussion}
\label{sec:discussion}

This paper classifies hardness transitions of QRAO optimization problems, showing that $\NP$-hard classical problems can be mapped to $\StoqMA$- or $\QMA$-hard optimization problems under QRAO compression. 
We call this change in completeness classification \textit{complexity amplification from compression}. 
QRAO compression can produce noncommuting Hamiltonian terms and hence can move the resulting global
optimization between $\NP$, $\StoqMA$, and $\QMA$ complexity regimes. The exact Hamiltonian
image classification identifies where these changes enter. Bipartite containment shows
that graph topology can place part of the positive interaction cone in both $\BQP$ and $\StoqMA$, and the order-robust compiler
construction then proves that $\QMA$-completeness survives exact compression
on a connected regular classical support graph and is robust to arbitrary fixed tie breaking. The density,
approximation, matching, and restricted state domain results expose complementary limits on the interplay of what compression and entanglement change. 
At the same time, our results are worst case and leave open the important question of when hardness arises in typical cases of QRAO, along with similar concerns for more general PCE compression and algorithms.

Qubit savings alone do not capture the computational consequences of compression, including the complexity of optimizing the resulting relaxation.

\paragraph{Where hardness enters in practice.}
The computationally hard task is global optimization over the chosen state domain, not energy evaluation of a fixed candidate state or decoding. 
$\QMA$- or $\NP$-hardness is therefore a worst-case barrier
to exact optimization of the compressed Hamiltonians, not general evidence for obstructed (or advantageous) quantum algorithms that use them. 
A useful algorithm may
avoid hard exact global optimization through restricted or locally optimized ansatz classes,
quantum or classical approximation algorithms,
or decoders tailored to particular input problem classes, among other techniques.

Complementary work proves that an end-to-end uniform polynomial-time quantum
or hybrid procedure preparing, on every MaxCut input and with inverse-polynomial
operational success, a state with encoded energy gain at least any fixed positive
fraction of the optimal classical gain above the random-cut baseline would imply
$\NP\subseteq\BQP$, in both the ordinary Ising and $d=2,3$ QRAO
encodings~\cite{HadfieldApproxQOpt2026}. In QRAO, the decoded mean gain equals
the encoded energy gain divided by~$d^2$. This includes the classical energy threshold highlighted above
and concerns uniform state preparation rather than our exact-optimization
classification over specified Hamiltonian state domains.

\paragraph{Related algorithmic and resource perspectives.}
Recent work combines QRAO with the quantum alternating operator ansatz
(QAOA)~\cite{Farhi14,Hadfield19}, including non-variational ansatz
constructions~\cite{He2025Nonvariational}, co-design of the Hamiltonian and
decoder~\cite{Suzuki2026Decoder}, and finite-budget QRAO--QAOA ensemble
analysis~\cite{SemreFrankel2026}. These studies address algorithm design and
performance, whereas our classification concerns the Hamiltonian optimization
problem produced by compression. Related complexity results show that exact or
exponentially precise expectation-value evaluation for QAOA with two or more
circuit layers is \#P-hard~\cite{Hadfield2026QAOAExpectation}. Bittel and
Kliesch show that classical parameter optimization can be $\NP$-hard even for
classically tractable underlying systems~\cite{BittelKliesch2021}. Those
results concern fixed-circuit evaluation or optimization within a fixed
ansatz, rather than optimization over specified state domains of a Hamiltonian
generated by compression.

Complementary resource bounds show how compressed encodings can move
cost into expectation-value geometry or readout overhead
~\cite{Hadfield2026NoFreeCompression}. These information-theoretic tradeoffs
are distinct from the complexity amplification classified here. Recent PCE work
also identifies ansatz--observable pairings whose encoded correlations are
efficiently classically computable~\cite{LizzioBoscoEtAl2026}. Algorithmic work
on the three-axis Quantum MaxCut problem studies partially entangled matchings
and QAOA-inspired variational ans\"atze on random regular
graphs~\cite{ApteEtAl2025,MarwahaSheSud2024}. Those results concern achievable
energies, whereas our matching and compiler results give exact state-domain
separations and worst-case complexity. Together, these comparisons reinforce
the need to distinguish the observable representation, admitted state family,
compiler, and optimization problem when assessing a compressed quantum
algorithm.

\paragraph{Implications for QRAO compilers.}
The hard connected compiler instances we construct
are dense and deliberately structured. 
The specified Qiskit version
builds its coloring graph from every nonzero quadratic coefficient, and so 
retains all filler edges in our constructions. A compression pipeline that deletes, thresholds, or contracts
small-weight edges before coloring
or considers alternative approaches altogether
may change the packing and so is
outside the coverage of Theorem~\ref{thm:connected}.

\paragraph{Compressed encodings and algorithms beyond QRAO.}
The recognition of QRAO as a $k=1$ PCE is an identification of the variable
encoding, not necessarily of the full optimization problem and workflow~\cite{SciorilliEtAl2025}. General PCE schemes may use higher locality Pauli strings, nonlinear loss functions, distinct quantum state preparation,
and sign-based or other decoding. The present completeness results do not automatically extend to those settings. Obtaining analogous complexity results for other compressed encodings with our approach would require an image of a hard problem family that preserves the promise gap when optimized over a suitable quantum state domain.

\section*{Acknowledgments}

{\emergencystretch=2pt\relax
The author thanks Filip Maciejewski and Davide Venturelli for helpful
discussions, and Chaithanya Rayudu for pointing out the recent results of Ref.~\cite{RayuduTakahashi2026}. This work was supported by the Air Force Research Laboratory under
Contract No.~FA8750-25-C-B0040.\par}

\section*{Code and data availability}

The code and data package accompanying this work is archived on
Zenodo~\cite{Hadfield2026QRAOCode}.

\section*{Author contributions and AI disclosure}

Stuart Hadfield is the sole author and takes responsibility for the research, manuscript, and accompanying code and data.
OpenAI ChatGPT assisted with literature synthesis, correctness checking of all results, and late-stage editing of the manuscript, in particular with 
explicit validation of all claims regarding Qiskit and NetworkX implementations along with generation of the accompanying code and data package.

\setlength{\bibsep}{0.5ex plus 0.2ex}
\bibliographystyle{plainnat}
\bibliography{bib}

\appendix

\section{Promise-gap and compiler details}
\label{app:promise}

\subsection{Promise-gap transformations}

Recall the setting of Theorem~\ref{thm:stoqma-transverse}. 
The reductions use two elementary forms of gap preservation. First, suppose the
source (input) minimum energy problem for $K$ has thresholds $a<b$ and gap
$\Delta=b-a$. If
$H=\alpha I-\beta K$,
 $\beta>0$,
then the corresponding maximum energy thresholds are
$\alpha-\beta a$ and $\alpha-\beta b$, with gap $\beta\Delta$. This covers the
axis lifts and the identity shifts arising from the MaxCut convention.
Second, if an implemented Hamiltonian is
$\widetilde H=H+\alpha'I+V$ with $\lVert V\rVert\leq\epsilon$, Weyl's
inequality (cf. Sec.~\ref{app:compiler-completion}) changes every extremal eigenvalue by at most $\epsilon$. Shifting
the thresholds by $\alpha'$ and reserving $2\epsilon$ of slack therefore
preserves a gap of at least $\beta\Delta-2\epsilon$. The transverse-field
rationalization and connected positive completion choose
$\epsilon\leq\beta\Delta/8$.
Finally, multiplicative approximation requires separate care because an identity shift
does not preserve value ratios. Theorem~\ref{thm:approx} therefore bounds the
norm of the complete filler contribution, including its identity part, relative
to the positive source threshold. No such extra step is needed for the exact
promise-problem classifications.

\subsection{Completion details for Theorem~\ref{thm:connected}}
\label{app:compiler-completion}

We complete the promise-gap and bit-complexity argument
omitted from the main-text proof.
Give every remaining clique or
matching edge a common positive weight $\varepsilon$, and let $M$ be the
number of these filler edges. Their Hamiltonian contribution is
\begin{equation}
 V=\frac{M\varepsilon}{2}I+V_0,
 \qquad
 V_0=-\frac{d\varepsilon}{2}
 \sum_{e\in E_{\mathrm{fill}}}P_eQ_e.
\end{equation}
Since every Pauli product has norm one, we have 
 $\lVert V_0\rVert\leq\frac d2M\varepsilon.$
Choosing $\varepsilon\leq\Delta/(4dM)$ then gives
$\lVert V_0\rVert\leq\Delta/8$. The identity coefficient in $V$ 
gives the shift
$\alpha=M\varepsilon/2$. Applying Weyl's inequality
~\cite[Corollary~III.2.6]{Bhatia1997MatrixAnalysis} bounds
eigenvalue shifts by the spectral norm of a perturbation \[
\left|
\lambda_{\max}(H_{\mathrm{out}}-\alpha I)
-\lambda_{\max}(H_{\mathrm{src}})
\right|
\leq \lVert V_0\rVert
\leq \Delta/8.
\]
After shifting the thresholds by $\alpha$, the YES and NO thresholds
can each move toward one another by at most $\Delta/8$. The remaining
promise gap is therefore at least
$\Delta-2(\Delta/8)=3\Delta/4$. Because $M=O(n^2)$ and $\Delta$ is
inverse polynomial, $\varepsilon$ and all resulting edge weights have
polynomial bit complexity. Altogether this proves the theorem.

\section{Cyclic-twist balance and local alignability}
\label{app:twists}

\begingroup

For the isotropic three-axis interaction, cyclic Pauli-frame choices admit an
exact graph description. Let $G=(V,E,w)$ be the interaction graph on the
compressed qubits, choose one orientation of each edge, and denote the
oriented edge set by $\vec E$. Set $P_0=X$, $P_1=Y$, and $P_2=Z$, with axis
indices taken modulo three. The three cyclic permutations are
$C_3=\{\mathrm{id},(XYZ),(XZY)\}$. Assign each $(u,v)\in\vec E$ a \emph{twist}
\[
t_{uv}\in\mathbb Z_3,\qquad t_{vu}=-t_{uv}.
\]
The twist $t_{uv}$ is the relative cyclic offset between the Pauli labels at
the two endpoints. Thus, $t_{uv}=0$ pairs like axes, while $t_{uv}=1$ and $2$
give, respectively,
\[
X_uY_v+Y_uZ_v+Z_uX_v
\quad\text{and}\quad
X_uZ_v+Y_uX_v+Z_uY_v.
\]
The collection $t=\{t_{uv}\}$ is the \emph{twist field}, with Hamiltonian
\begin{equation}
K(t)=\sum_{(u,v)\in\vec E}w_{uv}
\sum_{a\in\mathbb Z_3}P_{a,u}P_{a+t_{uv},v}.
\label{eq:twisted-isotropic}
\end{equation}
The convention $t_{vu}=-t_{uv}$ makes $K(t)$ independent of the chosen edge
orientations.

\begin{proposition}[Cycle criterion for cyclic twists]
\label{prop:twists}
The Hamiltonian $K(t)$ is locally unitarily equivalent to the aligned
interaction $K(0)$ if and only if the oriented twist sum vanishes modulo three
around every cycle. When this holds, the aligning unitaries can be chosen as
powers of
\begin{equation}
C=\frac12(I-iX-iY-iZ),\qquad
CXC^\dagger=Y,\quad
CYC^\dagger=Z,\quad
CZC^\dagger=X.
\label{eq:cyclic-clifford}
\end{equation}
In $O(|V|+|E|)$ time, one can return either the aligning Clifford frames or a
cycle with nonzero twist sum.
\end{proposition}

\begin{proof}
For $s\in\mathbb Z_3$, define the cyclic coupling matrix
\[
\Pi(s)_{ab}=
\begin{cases}
1,&b=a+s\pmod 3,\\
0,&\text{otherwise}.
\end{cases}
\]
Writing $\boldsymbol P_u=(P_{0,u},P_{1,u},P_{2,u})^{\mathsf T}$, the
interaction on $(u,v)$ is
$\boldsymbol P_u^{\mathsf T}\Pi(t_{uv})\boldsymbol P_v$, and
$\Pi(s)\Pi(r)=\Pi(s+r)$.

A single-qubit unitary acts on the Pauli vector through a rotation
$R_u\in\mathrm{SO}(3)$. If local unitaries map $K(t)$ to $K(0)$,
coefficient matching on each edge gives
\[
\Pi(t_{uv})=R_uR_v^{\mathsf T}.
\]
Multiplying these identities around a cycle cancels adjacent rotations and
gives
\[
\Pi\bigl(\sum_{(u,v)\,\text{on the cycle}}t_{uv}\bigr)=I.
\]
Since $\Pi(1)$ has order three, the oriented twist sum must vanish modulo
three. Thus, a cycle with nonzero twist sum rules out this equivalence under arbitrary local one-qubit unitaries.

Conversely, choose a root in each connected component and propagate labels
$k_u\in\mathbb Z_3$ along a spanning tree according to
$t_{uv}=k_v-k_u$. Vanishing cycle sums make these labels path-independent and
ensure the same relation on every non-tree edge. Because $C^{-k_u}$ sends
$P_a$ to $P_{a-k_u}$, conjugating vertex $u$ by $C^{-k_u}$ aligns every edge.
A graph traversal constructs the labels in linear time. If a non-tree edge
violates the relation, that edge and the corresponding tree paths produce a
cycle with nonzero twist sum.
\end{proof}

The obstruction holds even for signed aligned targets. After local
rotations, let $D_{uv}=R_u^{\mathsf T}\Pi(t_{uv})R_v$ be the normalized edge
matrix, with $D_{vu}=D_{uv}^{\mathsf T}$ in reverse. Alignment makes every
$D_{uv}$ diagonal and in $\mathrm{SO}(3)$, so it squares to $I$, as does any
product of such matrices. The product along a cycle is
$R_{u_0}^{\mathsf T}\Pi(\sum t_{uv})R_{u_0}$, which instead has order three
when the twist sum is nonzero. Thus no local one-qubit unitaries can align
a frustrated cycle, even with signed couplings.

The vanishing-cycle-sum condition is the standard balance condition for a
$\mathbb Z_3$ gain graph (a graph whose oriented edges carry
$\mathbb Z_3$ labels)~\cite{Zaslavsky1989}. We call a twist field
\emph{balanced} when every cycle sum vanishes and \emph{frustrated} otherwise.

\endgroup

\paragraph{Cycle examples.}
On an oriented triangle, twists $(1,1,1)$ are balanced, whereas $(1,1,0)$
are frustrated. A four-cycle with one twist-$1$ edge and three twist-$0$
edges is bipartite but frustrated. These are fixed-packing examples
in the image of Proposition~\ref{prop:positive-cone}, with no claim about
greedy compilation of the same sparse graphs.

\subsection*{Mixed-axis source terms under greedy compilation}
On the hard source family below, the local Pauli permutations in
Theorem~\ref{thm:connected} make every source Pauli product mixed-axis
without changing the spectrum.

\begin{corollary}[Hardness with mixed-axis source terms]
\label{cor:mixed-axis-source}
Fix positive rational $J_X,J_Y,J_Z$ that are not all equal. For every fixed
stable degree-ordered greedy compiler, the compiled QRAO energy promise
problem remains $\QMA$-complete on connected, regular, non-bipartite,
exactly three-to-one compressed positive-weight MaxCut instances, even when
every non-filler source Pauli product has different axes at its two endpoints.
The nonidentity filler contribution has norm at most one eighth of the source
promise gap, as in Theorem~\ref{thm:connected}.
\end{corollary}

\begin{proof}
Piddock and Montanaro prove $\QMA$-completeness for the anisotropic interaction
$J_XX\otimes X+J_YY\otimes Y+J_ZZ\otimes Z$ with positive weights on
subgraphs of the triangular lattice~\cite[Theorem~4]{PiddockMontanaro2017}.
Give that lattice its standard three-sublattice coloring and assign the
cyclic Pauli frames $I$, $C$, and $C^2$ to the three color classes,
where $C$ is defined in Eq.~\eqref{eq:cyclic-clifford}. Every source edge
joins different frames. Each same-axis Pauli product therefore becomes
mixed-axis with its coefficient unchanged. Theorem~\ref{thm:connected}
supplies the positive completion and preserves the promise gap. Membership
in $\QMA$ follows from the two-local Hamiltonian formulation.
\end{proof}

The source twists sum to zero around every cycle, since their vertex-frame
offsets telescope. The mixed-axis condition concerns non-filler source terms.
Filler couplings may have aligned axes. The cited hardness theorem covers
anisotropic interactions only.

\section{Proofs for state-domain and approximation results} \label{appB:state-domain-proofs}

\subsection{Proof of Theorem~\ref{thm:states}}
\label{app:state-domain-proof}

\begin{proof}
We consider the QRAO optimal value promise problems for compressed Hamiltonians $\mathcal C_3$ over each of the four quantum state domains individually.
For QRAC codewords, Eq.~\eqref{eq:qrao-codeword-edge} equates the
energy of each encoded assignment with its witnessed classical MaxCut value. The codeword problem is therefore $\NP$-complete by the
classical part of Corollary~\ref{cor:main}. 

For pure product states, consider the unit weight
three-axis family whose product-state energy problem is $\NP$-complete
~\cite[Corollary~6]{KallaugherEtAl2025}. Applying
Theorem~\ref{thm:connected} with $d=3$ outputs the compressed Hamiltonian 
$H_{\mathrm{out}}
=
U_\gamma H_{\mathrm{src}}U_\gamma^\dagger+\alpha I+V_0$, 
$\lVert V_0\rVert\leq\Delta/8$,
where $\Delta$ is the source problem promise gap. Local unitaries preserve the
set of product states. Lemma~\ref{lem:stability} therefore
shows that, after accounting for the known shift $\alpha$, each
threshold changes by at most~$\Delta/8$. The remaining gap is at least
$3\Delta/4$, proving $\NP$-hardness on $\mathcal C_3$. Membership in $\NP$
follows because polynomial-precision single-qubit Bloch vectors provide
a witness whose energy can be efficiently evaluated to the required accuracy.

Next, for any Hamiltonian, a linear objective such as the energy $\operatorname{Tr}(H\rho)$ has the same maximum value over pure
product states as over their convex hull, whose elements are precisely the separable states.
Hence 
$\OPT_{\mathrm{sep}}(H)=\OPT_{\mathrm{prod}}(H)$, 
and so the above arguments for 
product states also give
$\NP$-completeness over the separable domain.

Finally, for unrestricted states, the claim follows as the $d=3$ case of
Corollary~\ref{cor:main}, which combines the
$\QMA$-complete three-axis family of
Theorem~\ref{thm:axis-classification} with the compiler realization of
Theorem~\ref{thm:connected}.

\end{proof}

\subsection{Proof of Theorem~\ref{thm:approx}}
\label{app:approx-proof}

\begin{proof}
For the first statement, Eq.~\eqref{eq:qrao-qmc} multiplies the Quantum MaxCut
value by three and therefore preserves every multiplicative ratio. Piddock's convention~\cite{Piddock2025QMC} uses $1/4$ per edge, so our Quantum MaxCut Hamiltonian and the QRAO lift are, respectively, twice and six times the Hamiltonian in that convention. These positive rescalings leave the ratio unchanged. In the
three disjoint copies, corresponding vertices have equal degree and identical
colored neighbor histories. 
 Under the vertex-major order, in which all three copies of each
original vertex are listed consecutively, (G1) assigns corresponding
vertices the same color, (G2) lists them consecutively, and (G3) packs
them onto the three Pauli axes of one qubit.
For the connected statement, apply Theorem~\ref{thm:connected} to the same
hard family, choosing the total filler contribution to have norm at most
$\delta$. If source YES and NO values are separated by thresholds $B>0$ and
$\eta B$, the completed YES value is at least $B-\delta$, while the completed
NO value is at most $\eta B+\delta$. Choose
$\delta/B<(1-\eta)/4$. Their ratio is then below
\begin{equation}
 \eta'=\frac{1+3\eta}{3+\eta}<1.
\end{equation}
All filler weights remain positive and have polynomial bit complexity, so the
resulting compressed Hamiltonian remains in $\mathcal C_3$ as claimed.
\end{proof}

\end{document}